\documentclass[11pt]{article}

\usepackage{amssymb,amsfonts,amsmath,amscd,dsfont,mathrsfs}
\usepackage{graphicx,float,psfrag,epsfig,subfigure,amsthm}
\usepackage{wrapfig}
\usepackage{booktabs, multirow}

\DeclareMathAlphabet{\mathpzc}{OT1}{pzc}{m}{it}
\DeclareMathOperator{\diag}{diag}

\usepackage{setspace}
\usepackage{amsmath,amssymb}
\usepackage{url}
\usepackage{enumerate}

\newtheorem{claim}{Claim}[section]
\newtheorem{lemma}[claim]{Lemma}

\newtheorem{conjecture}[claim]{Conjecture}
\newtheorem{theorem}{Theorem}

\newtheorem{definition}[claim]{Definition}

\newcommand{\ra}[1]{\renewcommand{\arraystretch}{#1}}
\def\<{\langle}
\def\>{\rangle}

\def\eps{{\varepsilon}}

\def\Esp{E^{\rm sp}}
\def\Et{E^{\rm t}}
\def\prob{{\mathbb P}}
\def\integers{{\mathbb Z}}
\def\naturals{{\mathbb N}}
\def\E{{\mathbb E}} 
\def\Var{{\textup{Var}}}

\def\const{C}

\def\Ev{{\sf E}}
\def\Re{\operatorname{Re}}

\def\one{\mathbf{1}}

\def\reals{\mathbb{R}}

\def\Z{{\mathbb{Z}}}
\def\const{{\kappa}}

\def\normal{{\sf N}}
\def\ve{\varepsilon}

\def\cS{{\cal S}}
\def\cA{{\cal A}}

\def\Var{{\rm Var}}
\def\cL{{\cal L}}

\def\de{{\rm d}}
\def\oD{\overline{B}}
\def\oF{\overline{F}}
\def\oW{\overline{W}}
\def\oY{\overline{Y}}
\def\oZ{\overline{Z}}
\def\oU{\overline{U}}

\def\Ftail{{\sf F}}

\def\etot{\varepsilon_{\rm tot}}
\def\K{{\cal K}}
\def\cR{\mathcal{R}}

\def\S{\mathcal{S}}
\def\Var{{\rm Var}}
\def\Ev{{\sf E}}
\def\ind{\mathbb{I}}

\newcommand\B{\rule[-1ex]{0pt}{0pt}}

\newcommand{\bd}{\mathbf}

\def\F{{\mathcal F}}
\def\Z{{\mathcal Z}}
\def\B{{\mathcal B}}
\def\W{{\mathcal W}}
\def\Y{{\mathcal Y}}
\def\one{{\bf 1}}

\def\hV{{\widehat{V}}}
\def\hE{{\widehat{E}}}
\def\hF{{\widehat{\F}}}
\def\hC{{\widehat{C}}}

\def\bb{{\bf b}}

\begin{document}

\title{Distributed Storage for Intermittent Energy Sources:\\ Control Design and Performance Limits}

\author{Yashodhan Kanoria,\; Andrea Montanari, \;\; David Tse, \; Baosen Zhang \\
\hspace{0.3cm} Stanford University \hspace{2.5cm} U.C. Berkeley
}

\maketitle

\begin{abstract}
One of the most important challenges in the integration of renewable
energy sources into the power grid lies in their `intermittent'
nature. The power output of sources like wind and solar varies with
time  and location due to factors that cannot be controlled by the
provider. Two strategies have been proposed to hedge against this
variability: 1) use energy storage systems to effectively average
the produced power over time; 2) exploit distributed generation to
effectively average production over location. We introduce a network
model to study the optimal use of storage and transmission resources
in the presence of random energy sources. We propose a
Linear-Quadratic based methodology to design control
strategies, and we show that these strategies are asymptotically
optimal for some simple network topologies. For these topologies,
the dependence of optimal performance on storage and transmission
capacity is explicitly quantified.

\end{abstract}

\section{Introduction}

It is widely advocated that future power grids should facilitate
the integration of a significant amount of renewable energy sources.
Prominent examples of renewable sources are
wind and solar. These differ substantially from traditional
sources in terms of two important qualitative  features:

\noindent {\bf Intrinsically distributed.} The power generated by these sources
is typically proportional to the surface occupied by the corresponding
generators. For instance, the solar power reaching ground is of the order of $2$ kWh
per day per  square meter. The wind power at ground level  is of the order
of  $0.1$ kWh per day per square meter \cite{MacKayBook}.
These constraints on renewable power generation have
important engineering implications. If a significant part of energy
generation is to be covered by renewables, generation is argued to be
distributed over large geographical areas.

\noindent {\bf Intermittent}. The output of renewable sources varies with  time
and  locations because of exogenous factors. For instance, in the case of wind and
solar energy, the power output is ultimately determined by
meteorological conditions.
One can roughly distinguish two sources of variability: \emph{predictable
variability}, e.g. related to the day-night cycle, or to seasonal
differences; \emph{unpredictable variability}, which is most
conveniently modeled as a random process.

Several ideas have been put forth to meet the challenges posed by
intermittent production. The first one is to leverage
geographically  distributed production. The  output of distinct
generators is likely to be independent or weakly dependent over large distances and
therefore the total production of a large number of well separated
generators should stay approximately constant,
by a law-of-large-number effect. This averaging effect should be
enhanced by the integration of different types of generators.

The second approach is to use energy storage to take
advantage of over-production at favorable times, and cope with
shortages at unfavorable times. Finally, a third idea is  `demand response',
which aims at scheduling in optimal ways some
time-insensitive energy demands. In several cases, this can be
abstracted as some special  form of energy storage  (for instance,
when energy is demanded for interior  heating, deferring a demand is
equivalent to exploiting the energy stored as hot air inside the building).

These approaches hedge against the energy source variability by
averaging over location, or by averaging over time.
Each of them requires specific infrastructures:
a power grid with sufficient transmission capacity in the first case,
and sufficient energy storage infrastructure in the second one.
Further, these two directions are in fact intimately related. With
current technologies, it is unlikely that centralized energy storage
can  provide effective time averaging of
--say-- wind power production, in a renewables-dominated scenario.
In a more realistic scheme, storage is distributed at the consumer
level, for instance leveraging electric car batteries
(a scenario known as vehicle-to-grid or V2G).  Distributed storage
implies, in turn, substantial changes of the demand on the transmission system.

The use of storage devices to average out intermittent renewables
production is well established. A substantial research effort has
been devoted to its design, analysis and optimization (see, for
instance,
\cite{Korpaas2003a,Korpaas2003b,Jukka2005,Brown2008,Nyamdasha2010,Su2011}).
In this line of work, a large renewable power generator is typically
coupled with a storage system in order to average out its power
production. Proper sizing, response time, and efficiency of the
storage system are the key concerns.

If, however, we assume that storage will be mainly distributed, the key
design questions change. It is easy to understand that
both storage and transmission capacity will have a significant effect on the ability of the network to average out the energy source variability. For example, shortfalls at a node can be compensated by either withdrawals from local storage or extracting power from the rest of the network, or a combination of both. The main goal of this paper is to understand the optimal way of utilizing simultaneously these two resources and to quantify the impact of these two resources on performance. Our contributions are:
\begin{itemize}
\item a simple model capturing key features of the problem;
\item a Linear-Quadratic(LQ) based methodology for the systematic design of control strategies;
\item a proof of optimality of the LQ control strategies in simple network topologies such as the 1-D and 2-D grids and in certain asymptotic regimes.
\item a quantification of how the performance depends on key parameters such as storage and transmission capacities.
\end{itemize}
The reader interested in getting an overview of the conclusions
without the technical details can read Sections \ref{sec:Model} and
\ref{sec:Discussion} only. Some details are omitted and deferred to the
journal version of this paper.


\section{Model and Problem Formulation}
\label{sec:Model}

The power grid is modeled as a  weighted graph $G$ with vertices
(buses or nodes) $V$, edges (lines) $E$. Time is slotted and will be
indexed by $t\in\{0,1,2,\dots\}$. In slot $t$, at each node $i\in V$
a quantity of energy $E_{{\rm p},i}(t)$ is generated from a
renewable source, and a demand of a quantity $E_{{\rm
d},i}(t)$ is received. For our purposes, these quantities only
enter the analysis through the net generation $\Z_i(t) = E_{{\rm
p},i}(t)-E_{{\rm d},i}(t)$. Let $\Z(t)$ be the vector of $\Z_i(t)$'s.
We will assume that $\{\Z(t)\}$ is a stationary process.

In order to average the variability in the energy supply, the system
makes use of storage and transmission. Storage is fully distributed: each node $i \in V$ has a device that
can store energy, with  capacity
$S_i$. We assume that stored energy can be fully recovered when needed
(i.e., no losses). At each time slot $t$, one can transfer an amount of energy $\Y_i(t)$ to storage at node $i$. If we denote by $\B_{i}(t)$ the amount of stored energy at node $i$ just before the beginning of time slot $t$, then:
\begin{equation}
\label{eq:storage_evo} \B_i(0) = 0, \qquad \B_i(t+1) = \B_i(t) +
\Y_i(t)
\end{equation}
where $\Y_i(t)$ is chosen under the constraint that $\B_i(t+1) \in [0,S_i]$.

We will also assume the availability at each node of a fast generation
source (such as a spinning reserve  or backup generator) which allows covering
 up of shortfalls. Let $\W_i(t)$ be the energy obtained from such a
 source at node i at time slot $t$. We will use the convention that
 $\W_i(t)$  negative means that energy is consumed from the fast
 generation source, and positive  means energy is dumped.
The cost of using fast generation energy sources is reflected in the
steady-state performance measure:
\begin{equation}
\label{eq:fast} \ve_{\W} \equiv \lim_{t\rightarrow \infty}
\frac{1}{|V|}\sum_{i\in V}\E\{\big(\W_i(t)\big)_-\}\
\end{equation}
The net amount of energy injection at node $i$ at time slot $t$ is:
$$ \Z_i(t) - \Y_i(t) - \W_i(t).$$
These injections have to be distributed across the transmission network, and the ability of the network to distribute the injections and hence to average the random energy sources over space is limited by the transmission capacity of the network. To understand this constraint, we need to relate the injections to the power flows on the transmission lines. To this end,  we adopt a `DC power flow' approximation model \cite{DCpowerflow}. \footnote{Despite the name, `DC flow' is an approximation to the AC flow}

Each edge in the network corresponds to a transmission line which is purely
inductive, i.e. with susceptance $-jb_e$, where $b_e\in \reals_+$. Hence, the network is lossless.
Node $i\in V$ is at voltage $V_{i}(t)$, with all the voltages assumed to have
the same magnitude, taken to be $1$ (by an appropriate choice of units). Let $V_{i}(t) = \, e^{j \phi_{i}(t)}$
denote the (complex) voltage at node $i$ in time slot $t$.
If $I_{i,k}(t) = -jb_{ik}(V_i(t)-V_k(t))$ is the electric current
from  $i$ to $k$, the corresponding power flow is then
$\F_{i,k}(t) = \Re[V_i(t)I_{i,k}(t)^*]=\Re [j b_{ik} (1 -
e^{j(\phi_i(t)- \phi_k(t))})] = b_{ik} \sin(\phi_i(t) -
\phi_k(t))$, where $\Re [\cdot]$ denotes the real part of a complex
number.

The DC flow approximation replaces $\sin(\phi_i(t)- \phi_k(t))$ by
$\phi_i(t) - \phi_k(t)$ in the above expression. This is usually a good approximation
since the phase angles at neighboring nodes are typically
maintained  close to each other to ensure that the generators at the
two ends remain in step.
This leads to the following relation between angles and power flow
$\F_{i,k}(t) =  b_{ik} (\phi_i(t)- \phi_{k}(t))$.
In matrix notation, we have
\begin{equation} \label{eq:nablaphi}
\F(t) = \nabla \phi(t),
\end{equation}
where $\F(t)$ is the vector of all power flows, $\phi(t)=(\phi_1(t),\dots,\phi_n(t))$ and $\nabla$ is a $|E| \times |V|$ matrix. $\nabla_{e,i}=b_e$ if $e=(i,k)$ for some $k$, $\nabla_{e,i}=-b_e$ if $e=(k,i)$ for some $k$, and $\nabla_{e,i}=0$ otherwise.

Energy conservation at node $i$ also yields
$$ \Z_i(t)-\Y_i(t)- \W_i(t) = \sum_k \F_{(i,k)}(t) = \big(\nabla^T
\bb^{-1}\F(t)\big)_i,$$ where $\bb=\diag(b_e)$ is an $|E|\times|E|$
diagonal matrix. Expressing $\F(t)$ in terms of $\phi(t)$, we get
\begin{equation} \label{eq:angle}
\Z(t)-\Y(t)-\W(t) = -\Delta \phi(t),
\end{equation}
where $\Delta = -\nabla^T \bb^{-1}\nabla$ is a  $|V| \times |V|$ symmetric matrix where $\Delta_{i,k}=- \sum_{l:(i,l) \in E} b_{il}$ if $i=k$, $\Delta_{i,k}=b_{ik}$ if $(i,k) \in E$ and $0$ otherwise.
In graph theory, $\Delta$ is called the graph Laplacian matrix. In power engineering, it is simply the imaginary part of the bus admittance matrix of the network.  Note that if $b_e\ge 0$ for all edges $e$, then $-\Delta \succcurlyeq 0$
is positive semidefinite. If the network is connected (which we assume
throughout),
it has only one eigenvector with eigenvalue $0$, namely the vector
$\varphi_v= 1$ everywhere (hereafter we will denote this as the vector
$\one$). This fits the physical fact that if all phases are rotated by the same amount, the powers in the network are not changed.

With an abuse of notation, we denote by $\Delta^{-1}$ the matrix 
such that $\Delta^{-1}\one = -M\, \one$, 
and  $\Delta^{-1}$ is equal to the inverse of $\Delta$ on the
subspace orthogonal to $\one$. Here $M>0$ is arbitrary, and all of our
results are independent of this choice (conceptually, one can think of
$M$ as very large).
Explicitly, let $\Delta=-\bd{V} \alpha^2 \bd{V}^T$ be the eigenvalue decomposition
of $\Delta$, where $\alpha$ is a diagonal matrix with non-negative entries. Define
$\alpha^{\dag}$ to be the diagonal matrix with $\alpha^{\dag}_{ii} =
M$ if $\alpha_{ii}=0$ and $\alpha^{\dag}_{ii} = \alpha_{ii}^{-1}$
otherwise. Then $\Delta^{-1} = -\bd{V} (\alpha^{\dag})^2 \bd{V}^T$.

Since the total power injection in the network adds up to zero
(which must be true by energy conservation),  we can invert
\eqref{eq:angle} and obtain
\begin{equation}
\phi(t)= -\Delta^{-1} (\Z(t)-\Y(t)-\W(t))\, .
\end{equation}
Plugging this into \eqref{eq:nablaphi}, we have
\begin{equation}
\F(t)  =  -\nabla\Delta^{-1}\big(\Z(t)-\Y(t)-\W(t)\big)\, .
\label{eq:FlowConstruction}
\end{equation}


There is a capacity limit $C_e$ on the power flow along each edge
$e$; this capacity limit depends on the voltage magnitudes and the
maximum allowable phase differences between adjacent nodes, as well
as possible thermal line limits. We will measure violations of this
limit by defining
\begin{equation}
\label{eq:flow_vio} \ve_{\F}  \equiv \lim_{t\rightarrow \infty}
\frac{1}{|E|}\sum_{e\in E} \E\{(\F_e(t)-C_e)_++ (-C_e-\F_e(t))_+\}\, .
\end{equation}

We are now ready to state the design problem:

\emph{For the dynamic system defined by equations
(\ref{eq:storage_evo}) and (\ref{eq:FlowConstruction}), design a
control strategy which, given the past and present random renewable
supplies and the storage states,
$$
\big\{(\Z(t),\B(t));\, (\Z(t-1),\B(t-1)),  \dots, \big \}
$$
choose the vector of energies $\Y(t)$ to put in storage and the vector
of fast generations $\W(t)$ such that the sum
 $\etot \equiv \ve_{\F}+ \ve_{\W}$, cf. Eq.~(\ref{eq:fast})  and (\ref{eq:flow_vio}),  is minimized.}

\section{Linear-Quadratic Design}
\label{sec:LQG_design}
In this section, we propose a design methodology that is based on
Linear-Quadratic (LQ) control theory.
\subsection{The Surrogate LQ Problem}
The difficulty of the control problem defined above stems from both
the nonlinearity of the dynamics due to the hard storage limits and
the piecewise linearity of the cost functions giving rise to the
performance parameters. Instead of attacking the problem directly,
we consider a surrogate LQ problem where the hard storage limits are
removed and the cost functions are quadratic:
\begin{eqnarray}
B_i(0)  =  \frac{-S_i}{2}, \qquad B_i(t+1)  =  B_i(t) + Y_i(t)\, ,\\
F(t)   =   -\nabla\Delta^{-1}\big(Z(t)-Y(t)-W(t)\big)\, ,
\end{eqnarray}
with performance parameters:
\begin{eqnarray}
 \ve^{\rm surrogate}_{W_i} & = &  \lim_{t \rightarrow \infty}
\E\{\big(W_i(t)\big)^2\}, \quad i \in V\, ,\\
\ve^{\rm surrogate}_{F_e} & =& \lim_{t\rightarrow \infty}
\E\{(F_e(t))^2\}, \quad e \in E\, ,\\
\ve^{\rm surrogate}_{B_i} & = & \lim_{t\rightarrow \infty}
\E\{(B_i(t))^2\}\, .
\end{eqnarray}
The process $B_i(t)$ can be interpreted as the deviation of a
virtual storage level process from the midpoint $S_i/2$, where the
virtual storage level process is no longer hard-limited but evolves
linearly. Instead, we penalize the deviation through a quadratic
cost function in the additional performance parameters $\ve^{\rm surrogate}_{B_i}$.

The virtual processes $B(t)$, $F(t)$, $W(t)$, $Y(t)$ and $Z(t)$ are
connected to the actual processes $\B(t)$, $\F(t)$, $\W(t)$, $\Y(t)$
and $\Z(t)$ via the mapping (where $[x]^b_a : = \max(\min(x,b),a)$ for $a \le b$):
\begin{align}
\Z_i(t) &= Z_i(t)\, , \;\; \F_{e}(t) = F_{e}(t)\, ,\label{eq:FlowMatch}\\
\B_i(t) &= \left [\, B_i(t) + S_i/2 \, \right ]_0^{S_i}\, ,
\label{eq:Bmapping}\\
\Y_i(t) &= \B_i(t+1) - \B_i(t)\, ,\label{eq:cBi}\\
\W_i(t) &= W_i(t) + Y_i(t) - \Y_i(t)\, .\label{eq:wi}
\end{align}
In particular, once we solve for the optimal control in the
surrogate LQ problem, (\ref{eq:cBi}) and (\ref{eq:wi}) tell us what
control to use in the actual system. Notice that the actual fast
generation control provides the fast generation in the virtual
system plus an additional term that keeps the actual storage level
process within the hard limit. Note also
\begin{align}
\label{eq:slashW_bound} \W_i(t) \geq W_i(t) - (B_i(t) - S_i/2)_+ -
(-B_i(t) - S_i/2)_+\, .
\end{align}
Hence the performance parameters $\ve_{\F}$, $\ve_{\W}$ can be
estimated from the corresponding ones for the virtual processes.

Now we turn to solving the surrogate LQ problem. First we formulate
it in standard state-space form. For simplicity, we will assume
$\{Z(t)\}_{t\ge 0}$ is an i.i.d. process 
(over time).\footnote{The case of a process $\{Z(t)\}_{t\ge 0}$ with
  memory can be in principle studied within the same framework, by
  introducing a linear state space model for $Z(t)$ and correspondingly
  augmenting the state space of the control problem.}
Hence $X(t):=[F(t-1)^T, B(t)^T]^T$ is
the state of the system. Also, $U(t):= [Y(t)^T, W(t)^T]^T$ is the control and $R(t) :=
[X(t)^T,Z(t)^T]^T$ is the observation vector available to the
controller. Then
\begin{eqnarray} \label{eq:state}
X(t+1) & = & \bd{A} X(t) + \bd{D} U(t) + \bd{E} Z(t),\\
\label{eq:observer} R(t) &=  & \bd{C} X(t) + \zeta(t)\, ,
\end{eqnarray}
where
\begin{eqnarray*}
\bd{A} \equiv \begin{bmatrix} 0 & 0 \\ 0 & \bd{I} \end{bmatrix},
\;\;\bd{D}\equiv \begin{bmatrix} -\nabla\Delta^{-1} & -\nabla\Delta^{-1} \\
  \bd{I} & 0 \end{bmatrix}\, ,\;\; \bd{E}\equiv \begin{bmatrix}
\nabla\Delta^{-1} \\ 0 \end{bmatrix}\, .
\end{eqnarray*}
and $\bd{C}=\begin{bmatrix} 0 & 0 \\ 0 & \bd{I} \end{bmatrix}$ and
$\zeta(t)=\begin{bmatrix} \bd{I} \\ 0 \end{bmatrix} Z(t)$.
We are interested in trading off between the performance parameters
$\ve_{F_e}, \ve_{W_i}$ and $\ve_{B_i}$'s. Therefore we introduce
weights $\gamma_e$'s , $\xi_i$'s, $\eta_i$'s and define the
Lagrangian
\begin{align}
\cL(t)&\equiv  \sum_{e=1}^{|E|} \gamma_e  \mathbb{E}\{F_e (t)^2\}+
\sum_{i=1}^{|V|} \xi_i \mathbb{E}\{B_i (t)^2\}+ \sum_{i=1}^{|V|}
\eta_i \mathbb{E} \{ W_i (t)^2\}
\nonumber\\
&=\E\big\{X(t)^T \bd{Q}_1 X(t) + U(t)^T \bd{Q}_2 U(t)\big\}\,
,\label{eq:cost}
\end{align}
where $\bd{Q}_1=\diag(\gamma_1,\dots,\gamma_{|E|},\xi_1,\dots,\xi_{|V|})$ and $\bd{Q}_2=\diag(0,\dots,0,\eta_1,\dots,\eta_{|E|})$.

We will let $\E\{Z(t)\}=  \oZ$.
We will also assume that $\Sigma_{Z}\equiv \E[Z(t)^T
Z(t)] = \bd{I}$, since if not, then we can define
$\bd{E}=[\nabla \Delta^{-1} \sqrt{\Sigma_{Z}}^{\,-1} 0]^T$, where $\sqrt{\Sigma_{Z}}$ is the symmetrical square root of $\Sigma_{Z}$.

An \emph{admissible control policy} is a mapping
$\{R(t),R(t-1),\dots,R(0)\}\mapsto U(t)$. The surrogate LQ problem
is defined as the problem of finding the mapping that minimizes the
stationary cost $\cL \equiv\lim_{t\to\infty}\cL(t)$.

Notice that the energy production-minus-consumption $Z(t)$ plays the
role both in the evolution equation (\ref{eq:state}) and the
observation (\ref{eq:observer}).
The case of correlated noise has been considered and solved for
general correlation structure in \cite{Kwong91}. Let $\bd{G}=
\mathbb{E} [ \zeta(t) Z(t)^T] = [\bd{I} \, 0 ]^T$, $R_1 (t) = [(Z(t)-\oZ)^T
, 0]^T$ and $R_2(t)=[ 0 , B(t)^T]^T$. Adapting the general result in \cite{Kwong91} to our special case, we have
\begin{lemma} \label{lem:lqg}
The optimal linear controller for the system in \eqref{eq:state} and \eqref{eq:observer} and the cost function in \eqref{eq:cost} is given by
\begin{equation}
U(t)=-(\bd{L} R_1(t) + \bd{K}^{-1} \bd{D}^T \bd{S} \bd{E} \bd{G}^T \bd{M}^{-1} R_2(t)) + \oU,
\end{equation}
where, letting $\eta \equiv \diag(\eta_1, \ldots, \eta_{|V|})$ and
$\gamma \equiv \diag(\gamma_1.\ldots,\gamma_{|V|})$:
\begin{eqnarray}
\oU & = & \left [\begin{array}{c} \oY \\ \oW \end{array} \right ] \nonumber \\
 & = & \left [\begin{array}{c} 0 \\ \left [I-\Delta(\nabla^T\gamma\nabla)^{-1}\Delta\eta\right]^{-1}\oZ
\end{array} \right], \label{eq:ubar}
\end{eqnarray}
and $\bd{S}$ is given by the algebraic Riccati equation
\begin{equation} \label{eq:S}
\bd{S}=\bd{A}^T \bd{S} \bd{A}+ \bd{Q}_1^T \bd{Q}_1 - \bd{L}^T \bd{K}\bd{L},
\end{equation}
where
$\bd{K} = \bd{D}^T \bd{S} \bd{D} + \bd{Q}_2^T \bd{Q}_2$,
$\bd{L} = \bd{K}^{-1} (\bd{D}^T \bd{S} \bd{A}+ \bd{Q}_2^T \bd{Q}_1) \label{eq:SL}$,
and
\begin{equation} \label{eq:M}
\bd{M}= \bd{C} \bd{J} \bd{C}^T+  \begin{bmatrix} \bd{I} & 0 \\ 0 &
  0 \end{bmatrix}\, ,
\end{equation}
where $\bd{J}$ satisfies the algebraic Riccati equation
$\bd{J}=\bd{A} \bd{J} \bd{A}^T + \bd{E} \bd{E}^T - \bd{O} \bd{M} \bd{O}^T$,
and
$\bd{O} = (\bd{A} \bd{J} \bd{C}^T + \bd{E} \bd{G}^T)\bd{M}^{-1}$.
\end{lemma}
Note that the optimal linear controller has a deterministic time-invariant component $\oU= [\oY^T,\oW^T]^T$ and an observation-dependent control $-(\bd{L} R_1(t)  + \bd{K}^{-1} \bd{D}^T \bd{S} \bd{E} \bd{G}^T \bd{M}^{-1} R_2(t)).$
 It is intuitive that $\oY =0$, since otherwise the storage process has a non-zero drift and will become unstable. The deterministic component $\oW$ for the fast generation can be seen to be the solution of a {\em static optimal power flow problem} with  deterministic net renewable generation $\oZ$ and cost function given by:
$$\cL = \sum_{e=1}^{|E|} \gamma_e  \oF_e ^2 + \sum_{i=1}^{|V|}
\eta_i  \oW_i^2.$$
On the other hand, the observation-dependent control is obtained by
solving the LQ problem with the net generations shifted to
zero-mean. Thus, the LQ design methodology naturally decomposes the
control problem into a static optimal power flow problem and a dynamic
problem of minimizing 
variances.

Notice that it might be also convenient to consider more general
surrogate costs in which a generic quadratic function of the means
$\oF_e$, $\oW_i$ is added to the second moment Lagrangian (\ref{eq:cost}).

\vspace{-0.3cm}
\subsection{Transitive Networks} \label{subsec:transitive_networks}
Lemma \ref{lem:lqg} gives an expression for the optimal linear controller. However, it  is difficult in
general to solve analytically  the Riccati equation. To
gain further insight, we consider the case of  transitive
networks.

An automorphism of a graph $G=(V,E)$ is a one-to-one mapping $f:V\to
V$ such that for any edge $e=(u,v)\in E$, we have
$e'=(f(u),f(v))\in E$. A graph is called \emph{transitive} if for
any two vertices $v_1$ and $v_2$, there is some automorphism $f: V
\rightarrow V$ such that $f(v_1)=v_2$. Intuitively, a graph is
transitive if it looks the same from the perspective of any of the
vertices. Given an electric network, we say the network is
transitive if it has a transitive graph structure, every bus has the
same associated storage, every line has the same capacity and
inductance, and $Z_i(t)$ is i.i.d. across the network. Without loss
of generality, we will assume $S_i=S$, $C_e=C$, $b_e=1$,
$\E[Z_i(t)]=\mu, \Var[Z_i(t)] =\sigma^2$. Since the graph is
transitive, it is natural to take the cost matrices as
$\bd{Q}_1=\diag(\gamma,\dots,\gamma,\xi,\dots,\xi)$ and $\bd{Q}_2=\diag(0,\dots,0,1,\dots,1)$.
Moreover, it can be seen from Eq.~\eqref{eq:ubar} that $\oU =0$. Since the mean net production is the same at each node, the static optimal power flow problem is trivial with the mean flows being zero. We are left with the dynamic variance minimizing problem.

Recall that $\Delta= -\bd{V} \alpha^2 \bd{V}^T$ is the eigenvalue
decomposition of $\Delta$. Since $\Delta =- \nabla^T \nabla$, the
singular value decomposition of $\nabla$ is given by $\nabla =
\bd{U} \alpha \bd{V}^T$ for some orthogonal matrix $\bd{U}$. The
basic observation is that, with these choices of $\bd{Q}_1$ and
$\bd{Q}_2$,  the Riccati equations diagonalize in the bases given by
the columns of $\bd{V}$ (for vectors indexed by vertices) and
columns of $\bd{U}$ (for vectors indexed by edges).

A full justification of the diagonal ansatz amounts to rewriting the
Riccati equations in the new basis.
For the sake of space we limit ourselves to deriving the optimal
diagonal control.
We rewrite the linear relation from $X(t)$ to $U(t)$  as
\begin{eqnarray}
Y(t)  & = & \bd{H}Z(t)-\bd{K}B(t)\,,\\
W(t) & = &\bd{P}Z(t)+\bd{Q}B(t)\, .
\end{eqnarray}
Substituting in Eq.~(\ref{eq:state}), we get
\begin{align}
B(t+1) &= (\bd{I}-\bd{K})B(t) + \bd{H}Z(t)\, ,\label{eq:FilteredEvol1}\\
F(t+1) &= \nabla \Delta^{-1}\big\{(\bd{I}-\bd{H}-\bd{P})Z(t)+
(\bd{K}-\bd{Q})B(t)\big\}\, ,\label{eq:FilteredEvol2}\\
W(t) &= \bd{P}Z(t)+\bd{Q}B(t)\, . \label{eq:FilteredEvol3}
\end{align}
Denoting as above by $\oD$, $\oF$, $\oW$ the average quantities, it is easy
to see that, in a transitive network, we can take  $\oF=0$,
$\oW=\mu$ and hence $\oD=0$. In words, since all nodes are
equivalent, there is no average power flow ($\oF=0$), the average
overproduction is dumped locally ($\oW=\mu$), and the average
storage level is kept constant ($\oD=0$).

We work in the basis in which $\nabla = \bd{U} \alpha \bd{V}^T$ is diagonal.
We will index singular values by
$\theta\in\Theta$
 hence $\alpha = \diag(\{\alpha(\theta)\}_{\theta\in\Theta})$
 (omitting hereafter the singular value $\alpha=0$
since the relevant quantities have vanishing projection along this
direction.)
In the examples treated in the next sections, $\theta$ will be a
Fourier variable.
Since the optimal filter is diagonal in this basis, we write  $\bd{K}= {\rm
  diag}(k(\theta))$, $\bd{H}= {\rm  diag}(h(\theta))$ and
$\bd{P} = {\rm  diag}(p(\theta))$, $\bd{Q}={\rm diag}(q(\theta))$.

We let $b_{\theta}(t)$, $z_{\theta}(t)$, $f_{\theta}(t)$,
$w_{\theta}(t)$ denote the components
of $B(t)-\oD$, $Z(t)-\mu$, $F(t)-\oF$, $W(t)-\oW$ along in the same basis.
From Eqs.~(\ref{eq:FilteredEvol1}) to (\ref{eq:FilteredEvol3}), we
get the scalar equations
\begin{eqnarray}
b_{\theta}(t)&= & (1-k(\theta))b_{\theta}(t-1) + h(\theta)z_{\theta}(t)\, ,\\
f_{\theta}(t) &=&
-\alpha^{-1}(\theta)\big\{(1-h(\theta)-p(\theta))z_{\theta}(t)+\nonumber \\
&& \phantom{-\frac{1}{\alpha(\theta)}\big\{(1}(k(\theta)-q(\theta))b_{\theta}(t-1)\big\}\,
,\\
w_{\theta}(t) &=&p(\theta)z_{\theta}(t)+q(\theta)b_{\theta}(t-1)\,
.
\end{eqnarray}
We will denote by $\sigma_B^2(\theta)$, $\sigma_F^2(\theta)$,
$\sigma_W^2(\theta)$ the stationary variances of $b_{\theta}(t)$,
$f_{\theta}(t)$, $w_{\theta}(t)$. From the above, we  obtain
\begin{align}
\sigma_B^2(\theta) & = \frac{h^2}{1-(1-k)^2}\,
\sigma^2\, ,\label{eq:VarFilter1}\\
\sigma_F^2(\theta) & =
\frac{1}{\alpha^2}\left[(1-h-p)^2+\frac{h^2(k-q)^2}
{1-(1-k)^2}\right]\,
\sigma^2\, ,\label{eq:VarFilter2}\\
\sigma_W^2(\theta) & = \left[p^2+\frac{h^2q^2}{1-(1-k)^2}\right]\,
\sigma^2\, . \label{eq:VarFilter3}
\end{align}
(We omit here the argument $\theta$ on the right hand side.)

In order to find $h,k,p,q$, we minimize the Lagrangian
(\ref{eq:cost}).
Using Parseval's identity, this decomposes over $\theta$, and we can
therefore separately minimize for each $\theta\in\Theta$
\begin{eqnarray}
\cL(\theta) = \sigma_W(\theta)^2+\xi\,\sigma_B(\theta)^2+\gamma\,
\sigma_F(\theta)^2\, .\label{eq:Lagrangian}
\end{eqnarray}
A lengthy but straightforward calculus exercise yields the following
expressions.
\begin{table*}[t] \centering \ra{1.8}
{\scriptsize
\begin{tabular}{@{$\,$}l@{ }| l@{ }l@{}| l@{ }l@{}| l@{ }l@{}}
\hline & \multicolumn{2}{c|}{No transmission ($C=0$)} &
\multicolumn{2}{c|}{No storage ($S=0$)} & \multicolumn{2}{c}{Storage
and Transmission}
\\\hline
\multirow{2}{*}{1-D} & & & $\Theta(\frac{\sigma^2}{C})$ & for $\mu C
< \sigma^2$ &
$\sigma \exp\!\left \{ - \sqrt{\frac{ CS}{\sigma^2}}\right \}^\dagger$ & for $\mu = \exp\! \left \{ - \omega\Big(\sqrt{\frac{ C S}{\sigma^2}}\Big)\right \}$\\
&$\Theta(\frac{\sigma^2}{S})$ & for $\mu S < \sigma^2\,$ & $\sigma
\exp\!\left \{ - \frac{\mu C}{\sigma^2}\right \} $ & otherwise  &
$\sigma \exp\!\left \{ - \frac{ C S}{\sigma^2}\right \}$& for $\mu = \exp\! \left \{ - o\Big(\sqrt{\frac{ C S}{\sigma^2}}\Big)\right \}$\\[3pt] \cline{1-1} \cline{4-7}
\multirow{2}{*}{2-D} & $\sigma \exp\!\left \{ - \frac{\mu
S}{\sigma^2}\right \} $ & otherwise & $\sigma \exp\!\left \{ -
\frac{ C }{\sigma}\right \}^\dagger$ & for $\mu = \exp\!\left \{ -
\omega\Big(\frac{ C }{\sigma}\Big)\right \}$ &
\multicolumn{2}{l}{\multirow{2}{*}{$\sigma \exp\!\left \{ - \frac{
C\max(C,S)}{\sigma^2}\right \}$
}}\\
&&& $\sigma \exp\!\left \{ - \frac{ C^2}{\sigma^2}\right \}$ & for $\mu = \exp\!\left \{ - o\Big(\frac{ C }{\sigma}\Big)\right \}$ & \\[3pt]\hline
\end{tabular}}
\caption{\small Asymptotically optimal $\eps_{\W} + \eps_{\F}$ in 1-D and
2-D grids. Logarithmic factors have been
neglected (also in the exponent). $\ \dagger$ indicates the lower
bound requires a conjecture in probability theory.}
\label{table:summary_results}
\end{table*}
\begin{theorem}
Consider a transitive network. The optimal linear control scheme is
given, in Fourier domain $\theta\in\Theta$, by
\begin{align}
p(\theta)  &= q(\theta) = \xi\, \frac{\sqrt{4\beta(\theta)+1}-1}{2}
\, ,
\label{eq:OptFilter1}\\
h(\theta) &= \frac{2\beta(\theta)+1-\sqrt{4\beta(\theta)+1}}{2\beta(\theta)}\, ,\label{eq:OptFilter2}\\
k(\theta)  &= \frac{\sqrt{4\beta(\theta)+1}-1}{2\beta(\theta)}\, ,\label{eq:OptFilter3}
\end{align}
where $\beta(\theta)$ is given by
\begin{eqnarray}
\beta(\theta) & = &\frac{\gamma}{\xi(\gamma+\alpha^2(\theta))}\, .\label{eq:BetaDef}
\end{eqnarray}
\end{theorem}
It is useful to point out a few analytical properties of these filters:
$(i)$ $\gamma/[\xi(\gamma+ d_{\rm max})]\le \beta\le 1/\xi$ with $d_{\rm
max}$ the maximum degree in $G$; $(ii)$ $0\le k\le 1$ is monotone decreasing as a function of $\beta$,
with $k = 1-\beta +O(\beta^2)$ as $\beta\to 0$ and $k= 1/\sqrt{\beta}
+O(1/\beta)$ as $\beta\to\infty$; $(iii)$
$0\le h\le 1$ is such that $h+k=1$. In particular, it is monotone increasing as a function of $\beta$,
with $h = \beta +O(\beta^2)$ as $\beta\to 0$ and $h= 1-1/\sqrt{\beta}
+O(1/\beta)$ as $\beta\to\infty$; $(iv)$ $p=q=\xi\beta k$.

\begin{theorem}
Consider a transitive network, and assume that the optimal LQ
control is applied. The variances are given as follows in terms of
$k(\theta)$, given in Eq.~(\ref{eq:OptFilter2}):
\begin{eqnarray}
\frac{\sigma_B^2(\theta)}{\sigma^2} & = &
\frac{(1-k(\theta))^2}{1-(1-k(\theta))^2}\, ,
\label{eq:OptSigmaB}\\
\frac{\sigma_F^2(\theta)}{\sigma^2} & = &
\frac{\alpha^2(\theta)}{(\gamma+\alpha^2(\theta))^2}\,
\frac{k^2(\theta)}{1-(1-k(\theta))^2}\, ,\label{eq:OptSigmaF}\\
\frac{\sigma_W^2(\theta)}{\sigma^2} & = &
\frac{\gamma^2}{(\gamma+\alpha^2(\theta))^2}\,
\frac{k^2(\theta)}{1-(1-k(\theta))^2}\, .\label{eq:OptSigmaW}
\end{eqnarray}
\end{theorem}

\section{1-D and 2-D Grids: Overview of Results}
\label{sec:Discussion}

For the rest of the paper, we focus on two specific network
topologies: the infinite one-dimensional grid (line network) and the
infinite two-dimensional grid. We will assume that the net generations
are independent across time and position, with common expectation
$\E Z_i(t) = \mu$, and we will place weak assumptions on the distributions (to be specified precisely later in Section \ref{sec:subgaussian}.)  We will focus on the regime when the
achieved cost is small. In Section \ref{sec:performance} we will
evaluate the performance of the LQ scheme on these topologies. In
Section \ref{sec:lower}. we will derive lower bounds on the
performance of any schemes on these topologies to show that the LQ
scheme is optimal in the small cost regime.  As a result, we
characterize explicitly the asymptotic performance in this regime.
The results are summarized in Table \ref{table:summary_results}.

Although the i.i.d. assumption simplifies significantly our
derivations, we expect that the qualitative features of our results
should not change for a significantly broader class of processes
$\{Z_i(t)\}_{t}$. In particular, we expect our results to generalize
under the weaker assumption that $Z_i(t)$
is stationary but close to independent beyond a time scale $T=O(1)$.

The parameter $\mu$, the mean of the net generation at each node, can be
thought of as a measure of the amount of  {\em over-provisioning}.
Let us first consider that case of a one-dimensional grid and assume
that $\mu$ is vanishing or negligible. In other words, the average
production balances the average load.
Our results imply that a dramatic improvement is achieved by a joint
use of storage and transmission resources. Consider first the case
$C=0$. The network then reduces to a collection of isolated nodes,
each with storage $S$.  It can be shown that the optimal cost decreases only slowly with
the  storage size  $S$, namely as $1/S$. Similarly, when there is only
transmission but no storage, the optimal
cost decreases only slowly with transmission capacity $C$, like
$1/C$. On the other hand, with both storage and transmission, the
optimal costs decreases {\em exponentially} with $\sqrt{CS}$.
Consider now positive over-provisioning  $\mu>0$.
When there is no storage, the only way to drive the cost significantly
down is at the expense of increasing the amount of
over-provisioning beyond $\sigma^2/C$.
The same performance can be achieved with a storage $S$ equalling to this amount of over-provisioning and with the actual amount of over-provisioning exponentially smaller.

\begin{figure}[b]
\centering
\includegraphics[scale = 0.7]{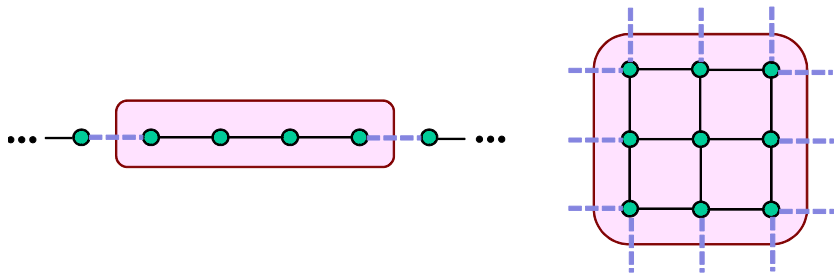}
\caption{Boxes of side $l$ in 1-D and in 2-D. The ratio $\text{boundary}/\sqrt{\text{volume}}$ is $\Theta(1/\sqrt{l})$ in 1-D, and $\Theta(1)$ in 2-D.}
\label{fig:1d_2d_cut}
\end{figure}


The 2-D grid provides significantly superior performance than the 1-D
grid. For example, the cost exponentially decreases with the
transmission capacity $C$ even without over-provisioning and without
storage. The increased connectivity in a 2-D grid allows much more
spatial averaging of the random net generations than in the 1-D grid.
In order to understand the fundamental reason for this difference,
consider the case of vanishing over-provisioning $\mu=0$ and vanishing
storage $S=0$  (also see Figure \ref{fig:1d_2d_cut}). Consider first a 1-D grid.  The aggregate
net generation inside a segment of $l$ nodes has variance $l \sigma^2$
and hence this quantity is of the order of \ $\sqrt{l} \sigma$. This
random fluctuation has to be compensated by power delivered from the
rest of the grid, but this power can only be delivered through the
two links, one at each end of the segment and each of capacity C.
Hence, successful compensation requires $l \lesssim C^2/\sigma^2$.
One can think of $l_* := C^2/\sigma^2$ as the {\em spatial scale}
over which averaging of the random generations is taking place. Beyond
this spatial scale, the fluctuations will have to be compensated by
fast generation. This fluctuation is of the order of $\sqrt{l_*}
\sigma / l_* = \sigma^2/C$ per node. Note that a limit on the spatial
scale of averaging translates to a large fast generation cost. In
contrast, in the 2-D grid,
$(i)$ the net generation, and $(ii)$ the total link
capacity connecting an $l \times l$  box to the rest of the
grid, both scale up linearly in $l$. This facilitates averaging over a
very large spatial scale $l$, resulting in a much lower fast
generation cost.

There is an interesting parallelism between the results for the 1-D
grid with storage and the 2-D grid without storage. If we set $S=C$,
the results are in fact identical. One can think of storage as
providing an additional dimension for averaging: time (Section
\ref{subsec:perflim_withstorage} formalizes this). Thus,
\emph{a one dimensional grid with storage behaves similarly to a
two-dimensional grid without storage.}

%
\section{Performance of LQ Scheme in Grids}
\label{sec:performance}

In this section we evaluate the performances of the LQ scheme on the
1-D and 2-D grids.  Both are examples of transitive graphs and hence
we will follow the formulation in Section
\ref{subsec:transitive_networks}. For these two examples, the operator $\Delta$ is in fact invariant to spatial shifts so the $\theta$-domain which diagonalizes the operator is simply the (spatial) Fourier domain.

For simplicity, in this section, we consider the case  where $Z_i(t)$ are gaussian. In the next section we show how all our results immediately generalize to a much larger and more realistic class of distributions.

Suppose that $Z_i(t) \sim \normal(\mu, \sigma^2)$ iid across nodes and time. It follows that $B,F,W$ are
Gaussian, and using Eq.~\eqref{eq:slashW_bound}, we get the
following estimates
\begin{eqnarray}
\ve_{\F} \leq 2\sigma_F\Ftail\Big(\frac{C}{\sigma_F}\Big)\, ,\;\;
\ve_{\W} \leq \sigma_B\Ftail\Big(\frac{S}{2\sigma_B}\Big)+\sigma_W\Ftail\Big(\frac{\mu}{\sigma_W}\Big) . \label{eq:Epsilon}
\end{eqnarray}
Here $\Ftail$ is the tail of the Gaussian distribution $\Ftail(z) \equiv \int_{z}^{\infty}\phi(x)\, \de x =
\Phi(-z)$, where $\phi(x) = \exp\{-x^2/2\}/\sqrt{2\pi}$ the
Gaussian density and $\Phi(x) = \int_{-\infty}^x\phi(u)\de u$ is the
Gaussian distribution.

In order to evaluate performances analytically and to obtain
interpretable expressions, we will focus on two specific regimes.
In the first one, no storage is available but large transmission
capacity exists. In the second, large storage and transmission
capacities are available.
%
%
\subsection{No storage}

In order to recover the performance when there is no storage, we let
$\xi\to\infty$, implying $\sigma_B^2\to 0$ by the definition of cost
function (\ref{eq:Lagrangian}). In this limit we have $\beta\to 0$,
cf. Eq.~(\ref{eq:BetaDef}). Using the explicit formulae for the
various kernels, cf. Eqs.~(\ref{eq:OptFilter1}) to
(\ref{eq:OptFilter3}), we get:
\begin{align*}
p,q =  \frac{\gamma}{\gamma+\alpha^2(\theta)} +O(1/\xi)\, ,\;
h  =  O(1/\xi)\, ,\;
k  =  1-O(1/\xi)\, .
\end{align*}
Substituting in Eqs.~(\ref{eq:FilteredEvol1}) to
(\ref{eq:FilteredEvol2})  we obtain
the following prescription for the controlled variables (in matrix notation)
\begin{eqnarray}
Y(t) = 0\, \;\;\;\;\; W(t)  = &  \gamma(-\Delta+\gamma)^{-1}Z(t)\,,
\end{eqnarray}
while the flow and storage satisfy
\begin{eqnarray}
B(t) =  0\, ,\;\;\;\;
F(t) = \nabla(-\Delta+\gamma)^{-1}Z(t)\, ,
\end{eqnarray}
The interpretation of these equations is quite clear. No storage is
retained ($B=0$) and hence no energy is transferred to storage.
The matrix $\gamma(-\Delta+\gamma)^{-1}$ can be interpreted a low-pass
filter and hence  $\gamma(-\Delta+\gamma)^{-1}Z(t)$ is a smoothing  of
$Z(t)$ whereby the smoothing takes place on a length scale $\gamma^{-1/2}$. The
wasted energy is obtained by averaging underproduction over regions
of this size.

Finally, using Eqs.~(\ref{eq:OptSigmaF}) and (\ref{eq:OptSigmaW}),
we obtain the following results for the variances in Fourier space
\begin{eqnarray*}
\frac{\sigma_F(\theta)^2}{\sigma^2}  =
\frac{\alpha^2(\theta)}{(\gamma+\alpha^2(\theta))^2}\, ,\;\;\;\;\;
\frac{\sigma_W(\theta)^2}{\sigma^2}  =
\frac{\gamma^2}{(\gamma+\alpha^2(\theta))^2}\, .
\end{eqnarray*}
%
%
\subsubsection{One-dimensional grid}

\begin{figure}
\centering
\includegraphics[scale = 0.7]{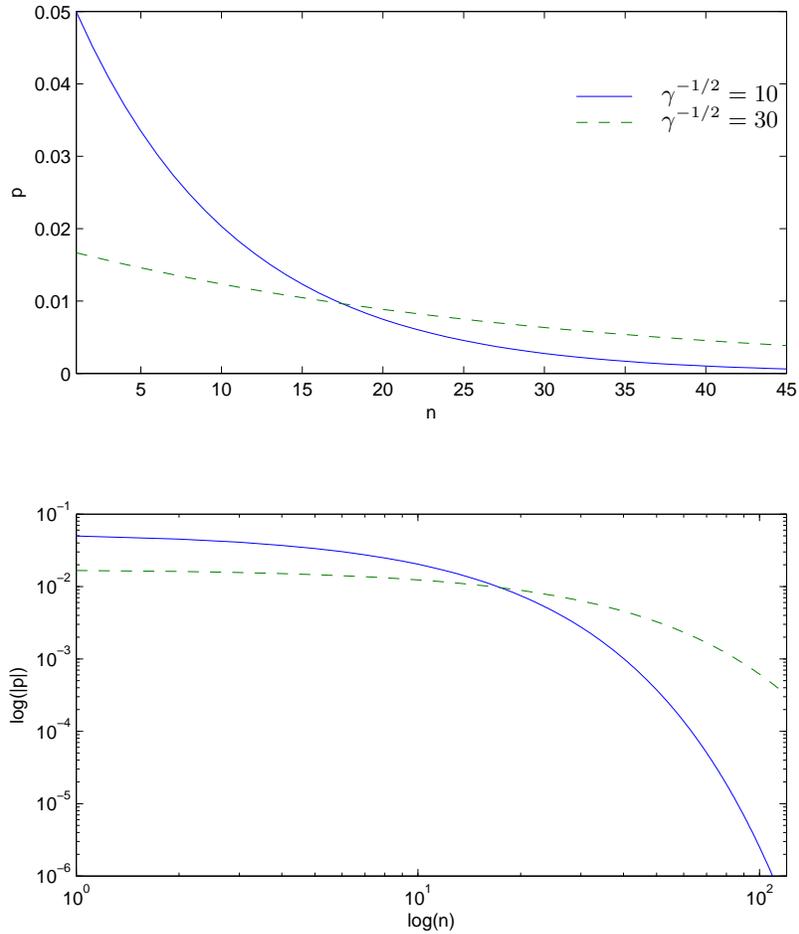}
\put(-80,338){{\footnotesize $\gamma^{-1/2} = 10$}}
\put(-80,328){{\footnotesize $\gamma^{-1/2} = 30$}}
\caption{The filter $\bd{P}$ for a one-dimensional grid $\integers$ in
  the case in which no storage is available. Notice that
  by translation invariance $\bd{P}_{i,j}= P_{i-j}$ for any
  $i,j\in\integers$, and further $P_{n} = P_{-n}$. Here we plot $P_n$
  for two values of the effective length scale $1/\sqrt{\gamma}$.}
  \label{fig:1DNoStorage}
\end{figure}

In this case $\theta\in  [-\pi,\pi]$, and
$\alpha(\theta)^2 = 2-2\cos\theta$ (the Laplacian $\Delta$ is
diagonalized via Fourier transform). The form of the optimal filter
$\bd{P}$ is shown in Figure \ref{fig:1DNoStorage}.

The Parseval integrals can be computed exactly but we shall limit
ourselves to  stating without proof their asymptotic behavior for small $\gamma$.
\begin{lemma}
For the one-dimensional grid, in absence of storage,
as $\gamma\to 0$, the optimal LQ control yields variances
\begin{align*}
\sigma_F^2 =
{\sigma^2} / {4\sqrt{\gamma}}\,\Big\{1+O(\gamma)\Big\}\, ,\;\;\;
\sigma_W^2 =
{\sigma^2\sqrt{\gamma}} / {4}\,\Big\{1+O(\gamma)\Big\}\, .
\end{align*}
\end{lemma}
Using these formulae and the equations (\ref{eq:Epsilon}) for the
performance parameters, we get the following achievability result.
\begin{theorem}\label{thm:1dNoStorage}
For the one-dimensional grid, in absence of storage, the optimal LQ control with Lagrange
parameter $\gamma = \mu^2/C^2$ yields, in the limit $\mu/C\to 0$, $\mu C/\sigma^2\to\infty$:
\begin{eqnarray}
\etot \le \exp\Big\{-\frac{2\mu\,
  C}{\sigma^2}\big(1+o(1)\big)\Big\}\, .\label{eq:1dNoStorage}
\end{eqnarray}
\end{theorem}
The choice of $\gamma$ given here is dictated by approximately
minimizing the cost. In words, the cost is exponentially small in the
product of the capacity, and overprovisioning $\mu C$. This is
achieved by averaging over a length scale $\gamma^{-1/2} = C/\mu$
that grows only linearly in $C$ and $1/\mu$. Note that the extent of averaging is limited by the transmission capacity $C$: the larger the extent of averaging, the larger the amount of power which has to be transported across the network. Optimal filters $\bd{P}$ for two different values of $\gamma^{-1/2}$ are displayed in Figure \ref{fig:1DNoStorage}.


%
%
\subsubsection{Two-dimensional grid}

In this case $\theta=(\theta_1,\theta_2)\in  [-\pi,\pi]^2$, and
$\alpha(\theta)^2 = 4-2\cos\theta_1-2\cos\theta_2$.
Again, we evaluate Parseval's integral as $\gamma\to 0$, and present
the result.
\begin{lemma}
For the two-dimensional grid, in absence of storage, as $\gamma\to 0$, the optimal LQ control yields variances
\begin{align*}
\sigma_F^2 =
\frac{\sigma^2}{4\pi}\,\Big\{\log\Big(\frac{1}{e\gamma}\Big)+O(\gamma)\Big\}\, ,\;
\sigma_W^2 =
\frac{\sigma^2\gamma}{4\pi}\,\Big\{1+O(\gamma)\Big\}\, .
\end{align*}
\end{lemma}
Using these formulae and the equations (\ref{eq:Epsilon}) for the
performance parameters, and approximately optimizing over $\gamma$,
we obtain the following achievability result.
\begin{theorem}
\label{thm:2D_noS_largemu}
For the two-dimensional grid, in absence of storage, the optimal LQ control with Lagrange
parameter $\gamma = (\mu^2/C^2)\log(C^2/\mu^2e)$ yields, in the limit
$\mu/C\to 0$,
$C^2/(\sigma^2 \log(C/\mu))\equiv  M\to \infty$:
\begin{align}
\etot \le
\exp\left\{-\frac{2\pi  C^2}{\sigma^2\log (C^2/\mu^2e)}
\big(1+o(1)\big)\right\}\,\, .\label{eq:2dNoStorage}
\end{align}
\end{theorem}
Notice the striking difference with respect to the one-dimensional
case, cf. Theorem \ref{thm:1dNoStorage}. The cost  goes
exponentially to $0$, but now overprovisioning plays a significantly
smaller role. For instance, if we fix the link capacity $C$ to be
the same, the exponents in Eq.~(\ref{eq:1dNoStorage}) are matched if
$\mu_{\rm 2d}\approx \exp(-\pi C/2\mu_{\rm 1d})\}$, i.e. an
exponentially smaller overprovisioning is sufficient.
%
%
\vspace{-0.2cm}
\subsection{With Storage}

In this section we consider the case in which storage is available.
Again we focus on the regime where the optimal cost is small. Within our LQ formulation we want therefore to penalize $\sigma_W$ much more than $\sigma_B$ and $\sigma_F$.
This corresponds to the asymptotics
$\gamma\to 0$, $\xi \equiv \gamma/s\to 0$ (the ratio $s$ need not to be
fixed).
It turns out that the relevant behavior is obtained by considering
$\alpha^2=\Theta(\gamma)$ and hence $\beta\to\infty$. The linear
filters are given in this regime by
\begin{align*}
p(\theta)=q(\theta)& =({\gamma}/{\sqrt{s}})\,
\big(\gamma+\alpha(\theta)^2\big)^{-1/2}\, ,\\
k(\theta) &\approx
({1}/{\sqrt{s}})\big(\gamma+\alpha(\theta)^2\big)^{1/2}\,
,\;\;\;\; h(\theta) \approx 1\, .
\end{align*}
Using these filters we obtain
\begin{eqnarray*}
\frac{\sigma_B(\theta)^2}{\sigma^2} &\approx &
\frac{1}{2}\left(\frac{s}{\gamma+\alpha(\theta)^2}\right)^{1/2}\, ,\\
\frac{\sigma_F(\theta)^2}{\sigma^2} &\approx &\frac{\alpha^2(\theta)}{2\sqrt{s}}
\left(\frac{1}{\gamma+\alpha(\theta)^2}\right)^{3/2}\, ,\\
\frac{\sigma_W(\theta)^2}{\sigma^2} &\approx &\frac{\gamma^2}{2\sqrt{s}}
\left(\frac{1}{\gamma+\alpha(\theta)^2}\right)^{3/2}\, .
\end{eqnarray*}

\subsubsection{One-dimensional grid}

\begin{figure}
\hspace{-1cm}
\centering
\vspace{-10pt}
\phantom{AAIi}\includegraphics[scale = 0.7]{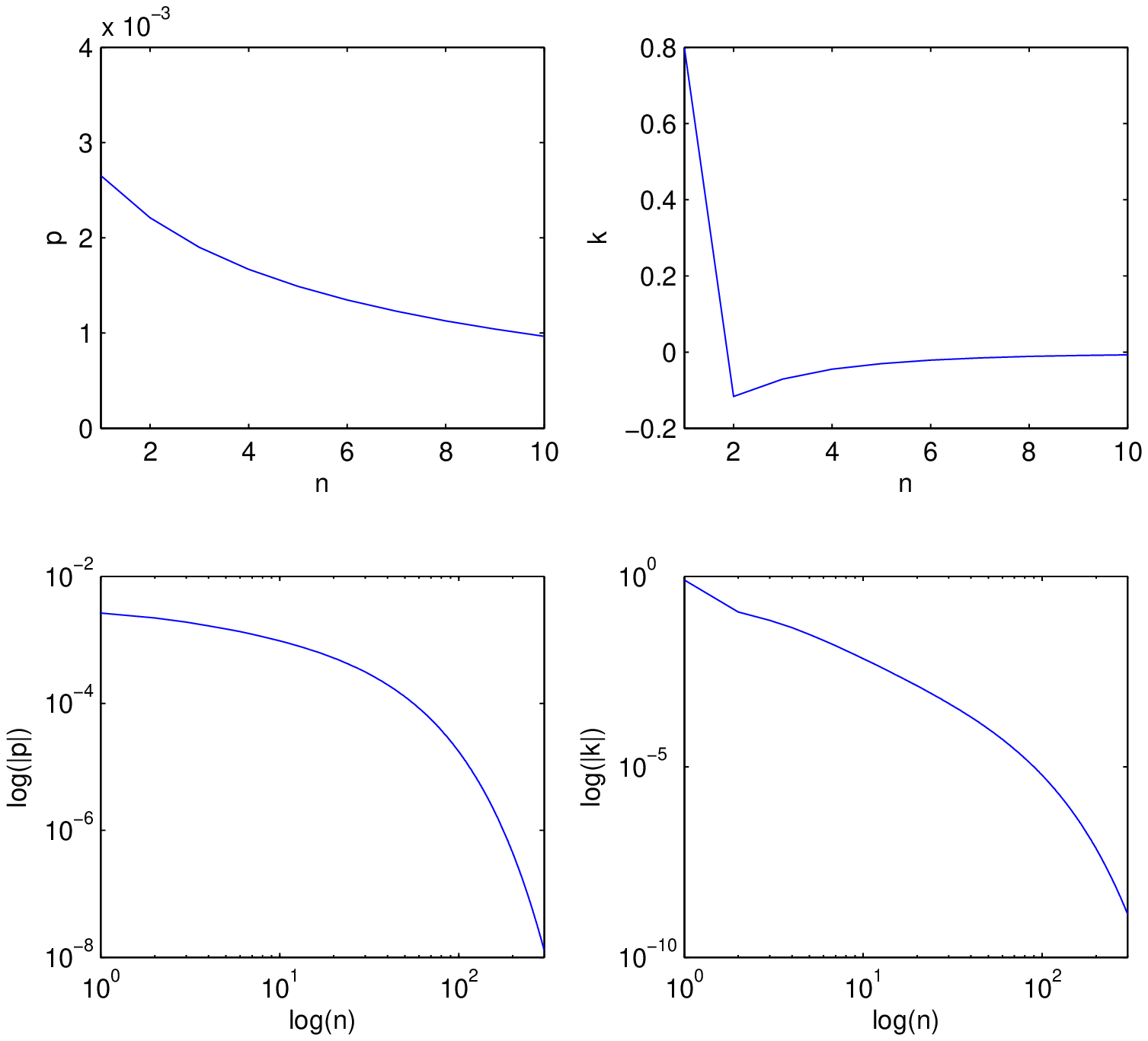}
\vspace{-15pt}
\includegraphics[scale = 0.7]{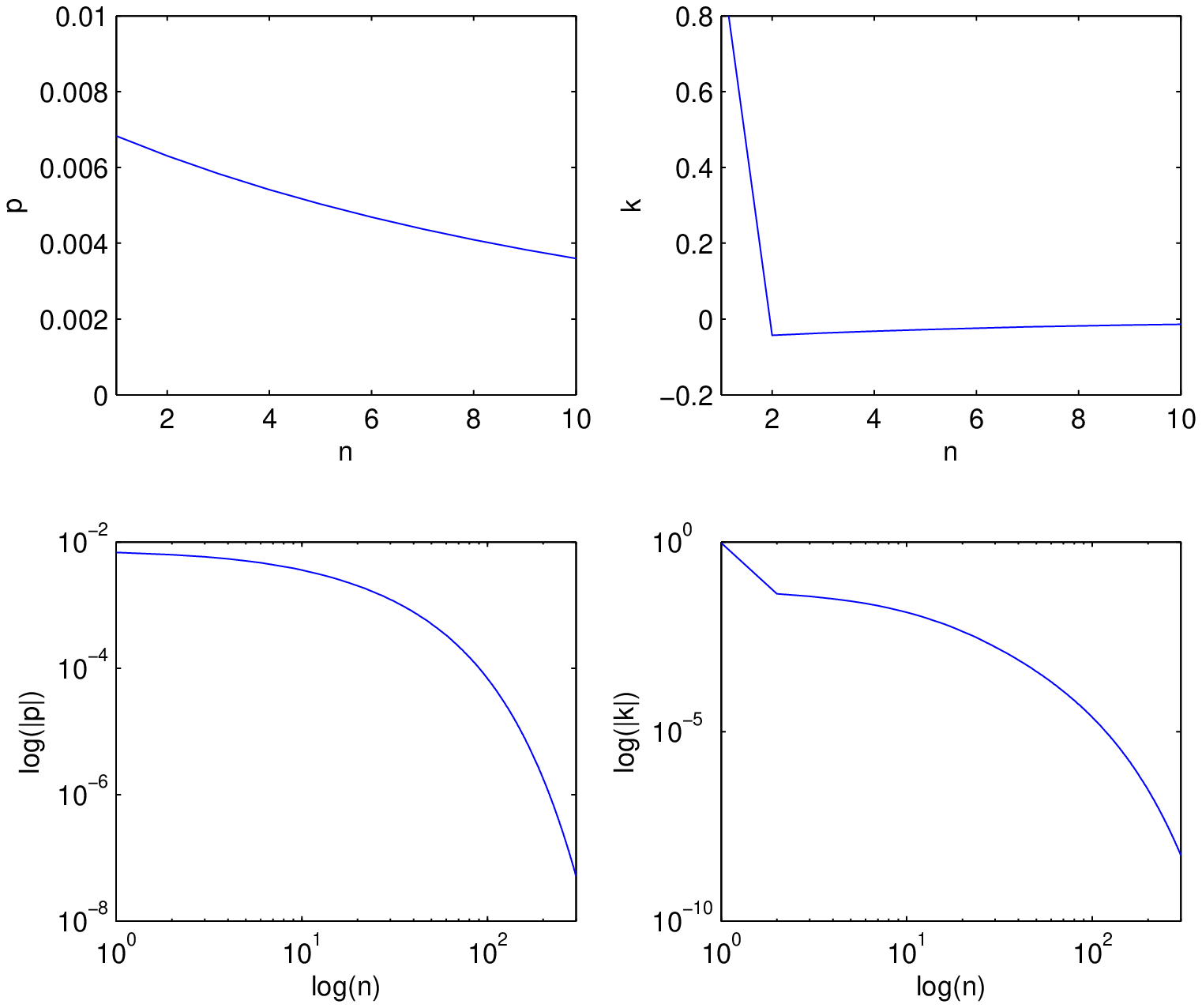}
\caption{The filters $\bd{P}$, $\bd{K}$ for a one-dimensional grid $\integers$ in
  the case in which storage is available. Notice that
  by translation invariance $\bd{P}_{i,j}= P_{i-j}$, $\bd{K}_{i,j}= K_{i-j}$ for any
  $i,j\in\integers$, and further $P_{n} = P_{-n}$, $K_{n} =
  K_{-n}$. Here we plot the filters for effective
  length scale: $1/\sqrt{\gamma} = 30$ and $S=C$ (top panels) or
  $S=C/4$ (lower panels).}
  \label{fig:1DStorage}
\end{figure}

The variances are obtained by Parseval's identity, integrating
$\sigma_{B,W,F}^2(\theta)$ over $\theta\in [-\pi,\pi]$. The form of
the optimal filters $\bd{P}$ and $\bd{K}$ is presented in Figure
\ref{fig:1DStorage}. We obtain the following asymptotic results.
\begin{lemma}
Consider a one-dimensional grid, subject to the LQ optimal control.
For $\gamma\to 0$ and $\xi=\gamma/s\to 0$
\begin{align*}
\frac{\sigma_B^2}{\sigma^2} & =
\frac{\sqrt{s}}{4\pi}\, \log\frac{1}{\gamma} +O(1) \, ,\;
\frac{\sigma_F^2}{\sigma^2}  =
\frac{1}{4\pi\sqrt{s}}\, \log\frac{1}{\gamma} +O(1,s^{-1}) \, ,\\
\frac{\sigma_W^2}{\sigma^2} & =
\frac{\Omega_1}{2\sqrt{s}}\, \gamma + O(\gamma^2,\gamma^{3/2}/s)\, ,
\end{align*}
where $\Omega_d$ is the integral
(here $\de^du\equiv \de u_1\times\cdots\times\de u_d$)
\begin{equation}
\Omega_d \equiv 1/(2\pi)^d
\int_{\reals^d}1/(1+\|u\|^2)^{3/2}\, \de^d u\, .\label{eq:OmegaDef}
\end{equation}
\end{lemma}

Using Eqs.~(\ref{eq:Epsilon}) to estimate the total cost $\etot$ and
minimizing it over $\gamma$ we obtain the following.
\begin{theorem}\label{thm:1dStorage}
Consider a one-dimensional grid and assume $CS/\sigma^2\to \infty$.
The optimal LQ scheme achieves the following performance:
\begin{align*}
\mu = e^{-\omega\Big(\sqrt{\frac{CS}{\sigma^2}}\Big)}
&\Rightarrow \etot\le \exp\Big\{-\sqrt{\frac{\pi
    CS}{2\sigma^2}} \big(1+o(1)\big)\Big\},\\
\mu = e^{-o\Big(\sqrt{\frac{CS}{\sigma^2}}\Big)}
&\Rightarrow  \etot\le \exp\Big\{-\frac{\pi
    CS(1+o(1))}{2\sigma^2\log C/\mu}\Big\} ,
\end{align*}
under the further assumption $\sqrt{\pi CS/2\sigma^2}
-\log(C/S)\to\infty$ (in the first case) and $\mu^2\log(C/\mu)/\min(C,S)^2\to
0$ (in the second).
In the first case the claimed behavior is achieved by $s=S^2/4C^2$, and
$\gamma = \exp\{-(2\pi C S/\sigma^2)^{1/2}\}$. In the second by
letting $s=S^2/4C^2$, and $\gamma = \mu^2\log (C/\mu)/(\pi\Omega_1 C^2)$.
\end{theorem}
This theorem points at a striking threshold phenomenon. If
overprovisioning is extremely small, or vanishing, then the cost is
exponentially small in $\sqrt{CS}$. On the other hand, even a modest
overprovisioning changes this behavior leading to a decrease that is
exponential in $CS$ (barring exponential factors). Overprovisioning
also reduces dramatically the effective averaging length scale
$\gamma^{-1/2}$. It also instructive to compare the second case in
Theorem \ref{thm:1dStorage} with its analogue in the case of no
storage, cf. Eq.~(\ref{eq:1dNoStorage}): storage seem to replace
overprovisioning.

\subsubsection{Two-dimensional grid}

As done in the previous cases, the variances of
$B$, $F$, $W$ are obtained by integrating
$\sigma^2_{B,F,W}(\theta)$ over  $\theta=(\theta_1,\theta_2)\in [-\pi,\pi]^2$.
\begin{lemma}
Consider a two-dimensional grid, subject to the LQ optimal control.
For $\gamma\to 0$ and $s=\Theta(1)$, we have
\begin{align*}
\frac{\sigma_B^2}{\sigma^2} & = G_B(s)
+O(1,\sqrt{\gamma}) \, ,\;\;\;\;\;\;
\frac{\sigma_F^2}{\sigma^2}  =  G_F(s)+O(\sqrt{\gamma}) \, ,\\
\frac{\sigma_W^2}{\sigma^2}  &=
\frac{\Omega_2}{2\sqrt{s}}\, \gamma^{3/2} + O(\gamma^{2})\, ,
\end{align*}
where $\Omega_2$ is the constant defined as per
Eq.~(\ref{eq:OmegaDef}), and $G_B(s)$, $G_F(s)$ are strictly positive
and bounded for $s$ bounded. Further, as $s\to\infty$
$G_B(s) = \K_2\sqrt{s}/2+O(1)$, $G_F(s) = \K_2/(2\sqrt{s})+O(1/s)$, where
$\K_2\equiv \int_{[-\pi,\pi]^2} \frac{1}{|\alpha(\theta)|}\;
\de\theta$.
\end{lemma}

Minimizing the total outage over $s, \gamma$, we obtain:
\begin{theorem}\label{thm:2DStorage}
Assume $CS/\sigma^2\to \infty$ and $C/S=\Theta(1)$. The optimal cost
for scheme a memory-one linear scheme on the two-dimensional grid
network then behaves as follows
\begin{eqnarray}
\etot \le \exp\Big\{-\frac{
    CS}{2\sigma^2\Gamma(S/C)} \big(1+o(1)\big)\Big\}\,  . \label{eq:2dStorage}
\end{eqnarray}
Here $u\mapsto \Gamma(u)$ is a function which is strictly positive
and bounded for $u$ bounded away from $0$ and $\infty$.
In particular, $\Gamma(u)\to \K_2$ as $u\to \infty$, and $\Gamma(u) =
\Gamma_0 u+o(u)$ as $u\to 0$ ($\Gamma_0>0$).

The claimed behavior is achieved by selecting $s=f(S/C)$, and
$\gamma$ as follows.
If $\mu =\exp\{-o(CS/\sigma^2) \}$ then $\gamma =\tilde{f}(S/C)(\mu^2/CS)^{2/3}$.
If instead $\mu = \exp\{-\omega(CS/\sigma^2) \}$, then
$\gamma =\exp\{-2CS/(3\Gamma(S/C)\sigma^2) \}$,
for suitable functions $f, \tilde{f}$
(In the first case, we also assume $\mu/C\to 0$.)
\end{theorem}
The functions $\Gamma, f, \tilde{f}$ in the last statement can be
characterized analytically, but we omit such characterization for
the sake of brevity. As seen by comparing with Theorem
\ref{thm:1dStorage}, the greater connectivity implied by a two
dimensional grid leads to a faster decay of the cost.

\section{Extension to a larger class of distributions}
\label{sec:subgaussian}

We find that our results from Section \ref{sec:performance} immediately generalize to a much broader class of  distributions for $Z_i(t)$ than Gaussian.

To define this class, first we provide the definition of {\em sub-Gaussian} random variables.
(See, for instance, \cite{Vershynin} for more details.)
\begin{definition}
A random variable $X$ is  \emph{sub-Gaussian} with  tail parameter
$s^2$ if, for any $\lambda\in\reals$,
\begin{eqnarray}
\E\big\{e^{\lambda(X-\E X)}\big\}\le \, e^{\lambda^2 s^2/2}\, .\label{eq:SubGaussianDef}
\end{eqnarray}
\end{definition}

Two important examples of sub-Gaussian random variables are:

\begin{enumerate}
\item Gaussian random variables with variance $\sigma^2$ are sub-Gaussian with tail parameter $s^2 = \sigma^2$.
\item Random variables with bounded support on $[a,b]$ are sub-Gaussian with tail parameter $s^2 = (b-a)^2/4$.
\end{enumerate}

Notice that the tail parameter is always an upper bound on the
variance $\sigma^2$, namely $\sigma^2\le s^2$ (this follows by Taylor
expansion of Eq.~(\ref{eq:SubGaussianDef}) for small $\lambda$).
We will consider the class of distributions for the net production $Z_i(t)$ to be sub-Gaussian  with tail parameter of the same order as the variance. More precisely:
\begin{definition}
Fix constant $\kappa > 0 $. Let $\S(\kappa)$ be the class of distributions such that the sub-Gaussian tail parameter $s^2$ and the standard deviation $\sigma^2$ satisfy:
$$ 1\le  \frac{s^2}{\sigma^2} < \kappa.$$
\end{definition}

This is a natural class of distributions for modeling renewable power production. The power generated by wind turbines, for example, are bounded between $0$ and a upper power limit, with significant probability that the power is near $0$ or capped at the upper limit. Hence, it is of bounded support with the range comparable to the standard deviation. It is sub-Gaussian with tail parameter of the same order as the variance.

We will now argue that all the results we derived in Section \ref{sec:performance} for Gaussian net productions extend to this class of distribution. The only fact we used to connect variances with the costs, where we used the Gaussianity assumption, is Eq.~\eqref{eq:Epsilon}. This equation implies that $\ve_{\F}$ decreases exponentially with $(C/\sigma_F)^2$, and $\ve_{\W}$ decreases exponentially with $(S/\sigma_B)^2$ and with $(\mu/\sigma_W)^2$, which in turns leads to Theorems \ref{thm:1dNoStorage},
\ref{thm:2D_noS_largemu}, \ref{thm:1dStorage}, \ref{thm:2DStorage}. We will show that these exponential dependencies  hold for distributions in $\S(\kappa)$ as well, and a similar versions of these theorems hold for these distributions.

First we need some elementary properties of sub-Gaussian random
variables.
The first property follows by elementary manipulations with moment
generating functions.
\begin{lemma}\label{lemma:CombiningSubGauss}
Assume $X_1$ and $X_2$ to be independent random variables with tail
parameters $s_1^2$ and $s_2^2$. Then, for any
$a_1,a_2\in\reals$, $X=a_1X_1+a_2X_2$ is sub-Gaussian with tail
parameter $(a_1^2 s_1^2+a_2^2 s_2^2)$.
\end{lemma}
Notice that by this lemma, the parameters of sub-Gaussian random
variables behave exactly as variances (as far as linear operations are involved).
In particular, it implies that the class $\S(\kappa)$ is closed under linear operations.


The second property is a well known consequence of Markov
inequality, and shows that the tail of a sub-Gaussian random variable
is dominated by the tail of a Gaussian with the same parameter.
\begin{lemma}\label{lemma:GaussianTails}
If $X$ is a sub-Gaussian random variable with parameter $s^2$,
then, for any $a\ge 0$ $\prob\{X\ge a+\E X \},\prob\{X\le - a + \E X
\}\le \exp\{-a^2/(2 s^2)\}$.

\end{lemma}

Now suppose the net productions $Z_i(t)$'s have distributions in $\S(\kappa)$.
The LQ scheme developed in Section \ref{sec:LQG_design} implies that
the controlled variables $B_i(t)$, $F_e(t)$, $W_i(t)$
are linear functions of the net productions $Z_i(t)$, and hence it follows that  $B_i(t)$, $F_e(t)$,
$W_i(t)$ are in $\S(\kappa)$. Now, if we let $F_e$ be the flow at edge $e$ at steady-state, with sub-Gaussian tail parameter $s_F^2$, then
\begin{eqnarray*}
\ve_{\F} & = & \E\{(\F_e-C)_++ (-C-\F_e)_+\} \\
& = & \int_{C}^\infty \prob\{F_e > a\} \de a + \int_{-\infty}^{-C}
\prob\{F_e < a\} \,\de a \\
& \le &  \int_{C}^\infty\exp\{-a^2/(2s_F^2)\}\, \de a +
\int_{-\infty}^{-C} \exp\{-a^2/(2s_F^2)\}\,\de a \\
& \le &  2 \int_{C}^\infty  \exp\{-a^2/(2 \kappa \sigma_F^2)\} \,\de a \\
& = & 2 \sqrt{2\pi} \sigma_F \,\Ftail\Big(\frac{C}{\sqrt{\kappa} \sigma_F}\Big).
\end{eqnarray*}
Here, as in Eq.~\eqref{eq:Epsilon}, $\Ftail(\,\cdot\, )$ is the
complementary cumulative distribution function of the standard
Gaussian random variable: $\Ftail(x) = 1-\Phi(x)$.  Similarly, one can show that:
$$ \ve_{\W} \le \sqrt{2\pi} \sigma_B\Ftail\Big(\frac{S}{2\sqrt{\kappa} \sigma_B}\Big)+ \sqrt{2\pi} \sigma_W\Ftail\Big(\frac{\mu}{\sqrt{\kappa} \sigma_W}\Big).$$

Thus, $\ve_{\F}$ decreases exponentially with $(C/\sigma_F)^2$, and $\ve_{\W}$ decreases exponentially with $(S/\sigma_B)^2$ and with $(\mu/\sigma_W)^2$, as in the Gaussian case, except for an additional factor of $1/\kappa$ in the exponent. This means that analogues of Theorems \ref{thm:1dNoStorage},
\ref{thm:2D_noS_largemu}, \ref{thm:1dStorage}, \ref{thm:2DStorage} also hold for distributions in $\S(\kappa)$ with an additional factor of $1/\kappa$ in the exponent of $\epsilon_{\rm tot}$. Note that all these exponents are proportional to $1/\sigma^2$, where $\sigma^2$ is the variance of the Gaussian distributed $Z_i(t)$. Therefore, equivalently, one can say that the performance under a sub-Gaussian distributed $Z_i(t)$ with tail parameter $s^2$ is at least as good as if $Z_i(t)$ were Gaussian with variance $s^2$. 

\section{Performance limits}
\label{sec:lower}

In this section, we prove general lower bounds on the outage $\etot = \ve_{\W}+ \ve_{\F}$
of \emph{any scheme}, on the 1-D and 2-D grids. Our proofs use cutset
type arguments.
Throughout this section, we will assume the $Z_i(t)$ to be i.i.d.
random variables, with  $Z_i(t)\sim\normal(\mu,\sigma^2)$, with the
exception of the case of a one-dimensional grid without storage,
cf. Theorem~\ref{thm:1dNoStorageLB}.
In this case, we will make the weaker assumption that the $Z_i(t)$ are
i.i.d. sub-Gaussian.

\subsection{No storage}

\subsubsection{One-dimensional grid}

\begin{theorem}\label{thm:1dNoStorageLB}
Consider a one-dimensional grid without storage, and assume the net
productions $Z_i(t)$ to be i.i.d. sub-Gaussian random variables in the
class $\S(\kappa)$. (In particular this assumption holds if
$Z_i(t)\sim\normal(\mu,\sigma^2)$.)

There exist finite constants $\const_0,\const_1,\const_3>0$ dependent
only on $\const$, such that the following happens.
For $\mu < \const_0\sigma$ and $\sigma< \const_0C$, we have
\begin{align}
\etot \geq
\left \{ \begin{array}{ll}
\const_1\sigma^2/C & \mbox{if } \mu < \sigma^2/C \, ,\\
\mu \exp \left\{ - \const_2 \mu C/\sigma^2 \right \}
& \mbox{otherwise.}
\end{array} \right .
\end{align}
\end{theorem}
\begin{proof}
Consider a segment of length $\ell$. Let $\Ev$ be the event that the
segment has net demand at least $3C$. Then we have
\begin{align}
\prob [\Ev] \geq \const_3 \exp\Big\{ -\frac{\const_4 (3C+\ell\mu)^2}{2\sigma^2 \ell} \Big\}\, ,
\end{align}
for some $\const_3,\const_4>0$. (This inequality is immediate for
$Z_i(t)\sim\normal(\mu,\sigma^2)$ and follows from Lemma
\ref{lemma:Tail} proved in the appendix for general random variables
in $\S(\kappa)$.)

If $\Ev$ occurs at some time $t$, this leads to a shortfall of at least $C$ in the segment of length $\ell$. This shortfall contributes either to $\ve_{\W}$ or to $2\ve_{\F}$, yielding
\begin{align}
2\etot \geq \ve_{\W}+2\ve_{\F} \geq \frac{\const_3 C}{\ell} \exp\Big\{
-\frac{\const_4 (3C+\ell\mu)^2}{\sigma^2 \ell} \Big\}\, .
\end{align}
Choosing $\ell = \min \left ( C/\mu, C^2/\sigma^2 \right )$, we obtain the result.
\end{proof}

Note that the lower bound is tight both  for $\mu \ge \sigma^2/C$ (by
Theorem \ref{thm:1dNoStorage}) and  $\mu <\sigma^2/C$ (by a simple
generalization of the same theorem that we omit).

\subsubsection{Two-dimensional grid}

We prove a lower bound almost matching the upper bound proved in Theorem \ref{thm:2D_noS_largemu}.

\begin{theorem}
There exists $\const < \infty$
such that, for $C\ge \min(\mu, \sigma)$,
$$
\etot \geq \sigma \exp\Big\{- \const C^2/\sigma^2\Big\}\, .
$$
\end{theorem}
\begin{proof}
Follows from a single node cutset bound.
\end{proof}
We next make  a conjecture in probability theory, which, if true, leads to a
significantly stronger lower bound for small $\mu$.
For any set of vertices $\cA$ of the two-dimensional grid, we denote
by $\partial \cA$ the \emph{boundary}  of $\cA$, i.e., the set of
edges in the grid that have one endpoint in $\cA$ and the other in
$\cA^{c}$.
\begin{conjecture}
There exists $\delta>0$ such that the following occurs for all $\ell \in
\naturals$.
Let $(X_v)_{v\in\cS}$ be a collection of i.i.d. $\normal(0,1)$ random
variables
indexed by $\cS=\{1,\dots,\ell\}\times\{1,\dots,\ell\}\subseteq \integers^2$.
Then
\begin{align}
\E \Big [\, \max_{\substack{\cA \subseteq \cS \textup{ s.t.}\\ |\partial \cA| \leq 4l}}\ \sum_{v \in \cA} X_{v}  \, \Big ] \geq \delta l \log l \, .
\label{eq:conj_symmetric}
\end{align}
\label{conj:large_sum_symmetric}
\end{conjecture}

It is not hard to see that this conjecture implies a tight lower bound.
\begin{theorem}
Consider the two-dimensional grid without storage, and assume Conjecture \ref{conj:large_sum_symmetric}. Then there exists $\const< \infty$ such that for any $\mu \leq \sigma\exp(- \const C/\sigma)$ and $C > \sigma$ we have
\begin{align}
\etot \geq \sigma \, \exp\Big\{-\const C/ \sigma\Big\} \, .
\end{align}
\label{theorem:2D_Conly_usingconj}
\end{theorem}
\begin{proof}
Consider a square of side $\ell$. Conjecture
\ref{conj:large_sum_symmetric} yields that we can find a subset of
vertices in the square with a boundary capacity no more than $4C\ell$,
but with a net demand of at least $\delta \sigma \ell \log \ell
- \mu \ell^2$.  This yields
\begin{align}
2\etot \geq \ve_{\W} + 2 \ve_{\F} \geq \frac{\delta \sigma \ell\log \ell - 4C\ell- \mu  \ell^2}{\ell^2} \; .
\end{align}
Choosing $\ell = \exp (\const C/\sigma)$ with an appropriate choice of $\const$, we obtain the result.
\end{proof}
\vspace{-0.3cm}

Our conjecture was arrived at based on a heuristic divide-and-conquer argument. We validated our conjecture numerically as follows:
We obtain a lower bound to the left hand side of Eq.~\eqref{eq:conj_symmetric}, by maximizing over a restricted class of subsets $\mathbb{S}_{\rm op}$, consisting of subsets that can be formed by dividing the square into two using an oriented path (each step on such a path is either upwards or to the right). It is easy to see that if $\cS \in \mathbb{S}_{\rm op}$, then $|\partial \cS | \leq 4 l$. Define
\begin{align}
G (l) \equiv \max_{\cS \in \, \mathbb{S}_{\rm op}} \  \sum_{v \in S} X_v \, ,
\label{eq:Gl_defn}
\end{align}
where $\mathbb{S}_{\rm op}$ is implicitly a function of $l$.
The advantage of considering this quantity is that $G(l)$ can be computed using a simple dynamic program of quadratic complexity.
Numerical evidence, plotted in Figure \ref{fig:conj_test}, suggests that $\E[G(l)] =\Omega(l \log l)$, which implies our conjecture.

\begin{figure}
\label{fig:conj_test}
\centering
\includegraphics[scale = 0.7]{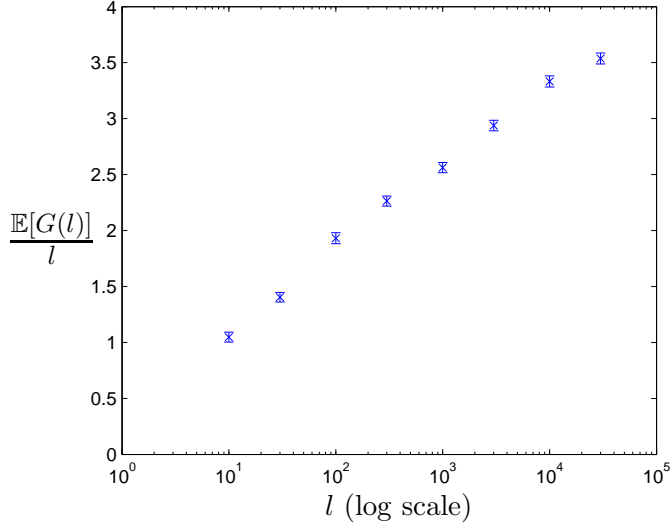}
\put(-150,0){$l$ (log scale)}
\put(-270,100){\Large $\frac{\E[G(l)]}{l}$}
\caption{Numerical evidence for Conjecture \ref{conj:large_sum_symmetric}: The approximate straight line in the plot (cf. definition of $G(l)$ in Eq.~\eqref{eq:Gl_defn}) suggests the validity of our conjecture.}
\end{figure}

\subsection{With storage}
\label{subsec:perflim_withstorage}
\subsubsection{One-dimensional grid}

Our approach involves mapping the time evolution of a control scheme in a one-dimensional grid, to a feasible (one-time) flow in a two-dimensional grid. One of the dimensions represents `space' in the original grid, whereas the other dimension represents time.

Consider the one-dimensional grid, with vertex set $\integers$.
We construct a two-dimensional `space-time' grid $(\hV, \hE)$
consisting of copies of each $v \in V$, one for each time
$t \in \integers$: define $\hV \equiv \{(v,t): v \in \integers, t \in \integers \}$.
The edge set $\hE$ consists of `space-edges' $\Esp$ and `time-edges' $\Et$.
\begin{align*}
\hE &\equiv \Esp \cup \Et\\
\Esp &\equiv \{((v,t), (v+1,t)): v \in \integers, t \in \integers \}\\
\Et &\equiv \{((v,t), (v,t+1)): v \in \integers, t \in \integers \}
\end{align*}
Edges are undirected. Denote by $\hC_e$ the capacity of $e \in \hE$. We define
$\hC_e \equiv  C$ for $e \in \Esp$ and $\hC_e=S/2$ for  $e \in \Et$.

Given a control scheme for the 1-D grid with storage, we define the
flows in the space-time grid as
\begin{align*}
\hF_{e} &\equiv \F_{(v, v+1)}(t) \qquad \qquad \mbox{for }e = ((v,t), (v+1,t)) \in \Esp\\
\hF_{e} &\equiv \B_v(t+1) - S/2 \qquad  \mbox{for }e = ((v,t), (v,t+1)) \in \Et
\end{align*}
Notice that these flows are not subject to Kirchoff constraints,
but the following energy balance equation is satisfied at each node $(v,t) \in \hV$,
\begin{align}
\Z_i(t) - \W_i(t) - \Y_{i}(t) = \sum_{(v',t') \in \partial (v,t)} \hF_{(v,t), (v',t')}
\end{align}
We use performance parameters as before (this definition applies to
finite networks and must be suitably modified for infinite graphs):
\begin{eqnarray*}
\ve_{\hF} &\! \equiv \!&\frac{1}{|\hE|}\sum_{e\in \hE} \E\{(\hF_e(t)-\hC_e)_++ (\hC_e-\hF_e(t))_+\}\, ,\\
\ve_{\W} &\! \equiv \!& \frac{1}{|\hV|}\sum_{(i,t)\in \hV}\E\{\big(W_i(t)\big)_-\}\,
.
\end{eqnarray*}
Notice that $\ve_\W$ is unchanged, and $\ve_\hF = \ve_\F$, in our mapping from the 1-D grid with storage to the 2-D space-time grid.

Our first theorem provides a rigorous lower bound which is almost tight
for the case $\mu = e^{-o(\sqrt{CS/\sigma^2})}$  (cf. Theorem \ref{thm:1dStorage}). It is proved by considering a rectangular region in the space-time grid of side $l = \max(C/S, 1)$ in space and $T= \max(1, S/C)$ in time.
\begin{theorem}
Suppose $\mu \leq \min(C, S)$, $CS/\sigma^2 > \max(\log(\sigma/\min(C,S)), 1)$. There exists $\const< \infty$ such that
\begin{align}
\etot \geq \sigma \exp(- \const CS/\sigma^2) \, .
\end{align}
\label{theorem:1D_CS_largemu_lb}
\end{theorem}
\begin{proof}
Consider a segment of length $\ell = \max(C/S, 1)$ and a sequence of $T = \ell S/C$ consecutive time slots. (Rounding errors are
easily dealt with.) The number of nodes in the corresponding region $\cR$ in the space-time grid is
$$n\equiv \ell T = \max(C,S)/\min(C,S).$$
The cut, i.e., the connection between $\cR$ and the rest of the grid, is of size $2(lS +TC) = 4 \max(C,S)$. The net generation inside $\cR$ is $\normal( n \mu, \sigma^2 n )$. Now $\mu \leq C$ by assumption, implying $n \mu \leq \max(C, S)$. Let $\Ev$ be the event that the net generation inside $\cR$ is at least $5 \max(C, S)$.
We have
$$
\prob[\Ev] \geq \exp\left(-  \frac{\const_1(\max(C,S))^2}{\sigma^2 n}\right) \geq \exp\left(- \frac{\const_1CS}{\sigma^2}\right)
$$
for some $\const_1< \infty$. Moreover, $\Ev$ leads to a shortfall of at least $\max(C,S)$ over $n$ nodes in the space-time grid. It follows that
$$
\etot \geq \big ( \max(C, S)/n \big ) \exp\left(- \frac{\const_1CS}{\sigma^2}\right)
= \min(C, S) \exp\left(- \frac{\const_1CS}{\sigma^2}\right)\, ,
$$
which yields the result, using $CS/\sigma^2 > \log(\sigma/\min(C,S))$.

\end{proof}

Next we provide a sharp lower bound for
small $\mu$ using Conjecture \ref{conj:large_sum_symmetric}.
Recall Theorem \ref{theorem:2D_Conly_usingconj} and notice that its
proof does not make any use of Kirchoff
flow constraints (encoded in Eq.~(\ref{eq:nablaphi})). Thus, the same result holds for a 2-D space-time grid. We immediately obtain the following result, suggesting that the  upper bound in Theorem \ref{thm:1dStorage} for small $\mu $ is tight.
\begin{theorem}
There exists $\const<\infty$ such that the following occurs if we assume that Conjecture \ref{conj:large_sum_symmetric} is valid. Consider the one-dimensional grid with parameters $C= S > \sigma$, and $\mu \leq \exp(-\const C/\sigma)$. We have
\begin{align}
\etot \geq \sigma \exp \Big \{ -\const \sqrt{CS/\sigma^2} \Big\} \, .
\end{align}
\end{theorem}

We remark that the requirement $C=S$ can be relaxed if we assume a generalization of Conjecture \ref{conj:large_sum_symmetric} to rectangular regions in the two-dimensional grid.

\subsubsection{Two-dimensional grid}
\begin{theorem}\label{theorem:2D_CS_lb}
There exists a constant $\const< \infty$ such that on the two-dimensional grid,
\begin{align}
\etot \, \geq  \sigma \exp\left\{  - \frac{\const C \max(C,S)}{\sigma_i^2} \right \}\, .
\end{align}
\end{theorem}
The theorem is proved by considering a single node, using a cutset type argument, similar to the proof of Theorem \ref{theorem:1D_CS_largemu_lb}. It implies that the upper bound in Theorem \ref{thm:2DStorage} is tight up to constants in the exponent.

\subsection*{Acknowledgements} This work was partially supported
by NSF grants CCF-0743978, CCF-0915145 and CCF-0830796.

\appendix

\section{A probabilistic lemma}
\label{sec:appendix}

\begin{lemma}
\label{lemma:Tail}
Let $\{X_1,X_2,\dots,X_n,\dots\}$ be a collection of
i.i.d. sub-Gaussian random variables in $\S(\kappa)$
with $\E X_1=0$, $\E\{X_1^2\}=\sigma^2$.

Then there exists finite constants $\kappa_1=\kappa_1(\kappa)>0$,
$\kappa_2=\kappa_2(\kappa)>0$,$n_0=n_0(\kappa)$ depending uniquely on
$\kappa$ such that, for all  $n\ge n_0$, $0\le\gamma\le
\kappa_1\sigma$, we have
\begin{eqnarray}
\prob\Big\{\sum_{i=1}^nX_i\ge \gamma n\Big\}\ge \frac{1}{4}\,
\exp\Big\{-\frac{n\gamma^2}{\kappa_2\sigma^2}\Big\}\, .
\end{eqnarray}
\end{lemma}
\begin{proof}
By scaling, we will assume, without loss of generality, $\sigma^2=1$.
Throughout the proof $\kappa',\kappa'',\dots$ denote constants
depending uniquely on $\kappa$. We  will use the same symbol even if
the constants have to be redefined  in the course of the proof.

For any $\lambda\in\reals$, let $\prob_{\lambda}$, $\E_{\lambda}$ denote
probability and expectation with respect to the measure defined
implicitly by
\begin{eqnarray}
\E_{\lambda}\{f(X_1,\dots,X_n)\}  \equiv
\frac{ \E\{f(X_1,\dots,X_n)\,
  e^{\lambda\sum_{i=1}^nX_i} \} }
{ \E\{e^{\lambda\sum_{i=1}^nX_i}\} }\, ,
\end{eqnarray}
for all measurable functions $f$. Notice that this measure is well
defined for all $\lambda$ by sub-Gaussianity.

Let $g(\lambda)\equiv \E_{\lambda}X_1$. Then $\lambda\mapsto
g(\lambda)$ is continuous, monotone increasing with $g(0)=0$,
$g'(\lambda) =\Var_{\lambda}(X_1)$, $g''(\lambda) = \E_{\lambda}
X_1^3-3\E_{\lambda}X_1\E_{\lambda} X_1^3$.
Bounding these quantities by sub-Gaussianity, it follows that, for $0\le
\lambda\le \kappa'$, we have
$1/\kappa''\le \Var_{\lambda}(X_1)\le \kappa''$, and hence
\begin{eqnarray}
\frac{\lambda}{\kappa''}\le g(\lambda)\le \kappa''\lambda\, .
\end{eqnarray}
Define $\gamma_{+/-}\equiv\lim_{\lambda\to\pm\infty} g(\lambda)$. Notice that
$\gamma_-<0<\gamma_+$ and that $g^{-1}$ (the inverse function of $g$)
is well defined on the interval $(\gamma_-,\gamma_+)$.

Define
\begin{eqnarray}
h(\lambda) \equiv\frac{(\E e^{\lambda X_1})^2}{\E (e^{2\lambda X_1})}\, .
\end{eqnarray}
By Taylor expansion, we get $h(\lambda) =
1-\lambda^2\E(X_1^2)+O(\lambda^3) =1-\lambda^2+O(\lambda^3)$.
Proceeding as above, it is not hard to prove that
$h(\lambda)\ge 1-\kappa''\lambda^2$ for all $0\le \lambda\le \kappa'$
for some finite constants $\const',\const''>0$ (eventually different
from above).
Finally, for $\gamma\in(\gamma_-,\gamma_+)$ we define
\begin{eqnarray}
H(\gamma)\equiv h(g^{-1}(\gamma))\, .
\end{eqnarray}
Combining the above, we have $H(\gamma) = 1-\kappa'' \gamma^2$ for
all $\gamma\in [0,\kappa']$, and  therefore
\begin{eqnarray}
H(\gamma)\ge e^{-\gamma^2/\kappa_2}\, \;\;\;\;\;\mbox{ for all }\gamma\in[0,\kappa_1]\, .
\end{eqnarray}

Now, for $\gamma\in [0,\kappa_1]$, let $\Ev=\Ev(\gamma)$ be the event that
$X_1+\dots+X_n\ge n\gamma$. Take $\lambda = g^{-1}(\gamma)$ and define
$Z(\lambda) \equiv \exp\{\lambda\sum_{i=1}^nX_i\}$. By Cauchy-Schwarz
inequality
\begin{eqnarray*}
\prob\{\Ev\}&\ge
&\frac{\E\{\ind_{\Ev}Z(\lambda)\}^2\}}{\E\{Z(\lambda)^2} \\
& = &\prob_{\lambda}\{\Ev\}^2 \;\frac{\{\E Z(\lambda)\}^2}{\E\{Z(\lambda)^2\}}\\
& = & \prob_{\lambda}\{\Ev\}^2 \, H(\gamma)^n\\
&\ge & \prob_{\lambda}\{\Ev\}^2\,
\exp\Big\{-\frac{n\gamma^2}{\kappa_2}\Big\}\, .
\end{eqnarray*}
The proof is completed by noting that $\prob_{\lambda}\{\Ev\}^2\ge 1/4$
for all $n\ge n_0(\kappa)$, by Berry-Esseen central limit theorem
(note indeed that, under $\prob_{\lambda}$, $X_1$,\dots,$X_n$ have mean $\gamma$, variance lower
bounded by $\Var_{\lambda}(X_i)\ge\kappa'>0$ and $\E_{\lambda}(|X_i|^3)\le \kappa''<\infty$).
\end{proof}

\bibliographystyle{IEEEtran}

\end{document}